\documentclass[letterpaper, 10 pt, conference]{ieeeconf}
\IEEEoverridecommandlockouts
\usepackage{graphics}
\usepackage{epsfig}
\usepackage{amsmath}
\usepackage{amssymb}
\usepackage{xcolor}
\let\labelindent\relax
\usepackage{enumitem}
\usepackage{multirow}
\usepackage{url}

\usepackage{amsthm}
\theoremstyle{definition}

\newtheorem{theorem}{\normalfont\bfseries Theorem}

\newtheorem{proposition}{\normalfont\bfseries Proposition}
\newtheorem{definition}{\normalfont\bfseries Definition}

\newtheorem{corollary}{\normalfont\bfseries Corollary}

\newtheorem*{observation}{\normalfont \bfseries Observation}

\newcommand{\R}{\mathbb{R}}
\newcommand{\X}{\mathbb{R}^n}
\newcommand{\U}{\mathbb{R}^m}
\renewcommand{\S}{\mathcal{S}}
\newcommand{\bS}{\partial \S}
\newcommand{\Se}{\mathcal{S}_{\rm e}}
\newcommand{\bSe}{\partial \Se}
\newcommand{\K}{\mathcal{K}}
\newcommand{\Kinf}{\mathcal{K}_\infty}
\newcommand{\Keinf}{\mathcal{K}_\infty^{\rm e}}
\newcommand{\grad}[1]{\nabla #1}
\renewcommand{\L}[2]{L_{#1} #2}
\newcommand{\kd}{k_{\rm d}}
\newcommand{\ks}{k_{\rm s}}
\newcommand{\he}{h_{\rm e}}
\newcommand{\alphae}{\alpha_{\rm e}}
\newcommand{\vL}{v_{\rm L}}
\newcommand{\aL}{a_{\rm L}}
\newcommand{\vmax}{v_{\rm max}}
\newcommand{\Dst}{D_{\rm st}}
\newcommand{\Dsf}{D_{\rm sf}}
\renewcommand{\TH}{T_{\rm h}}
\newcommand{\TTC}{T_{\rm c}}

\title{\LARGE \bf
On the Safety of Connected Cruise Control: \\
Analysis and Synthesis with Control Barrier Functions
}

\author{Tamas G. Molnar, G\'abor Orosz and Aaron D. Ames%
\thanks{*This research is supported in part by the National Science Foundation (CPS Award \#1932091), Dow (\#227027AT) and Aerovironment.}%
\thanks{T. G. Molnar and A. D. Ames are with the Department of Mechanical and Civil Engineering, California Institute of Technology, Pasadena, CA 91125, USA,
{\tt\small tmolnar@caltech.edu, ames@caltech.edu}.}%
\thanks{G. Orosz is with the Department of Mechanical Engineering and the Department of Civil and Environmental Engineering, University of Michigan, Ann Arbor, MI 48109, USA,
{\tt\small orosz@umich.edu}.}%
}

\begin{document}

\maketitle
\thispagestyle{empty}
\pagestyle{empty}

\begin{abstract}
Connected automated vehicles have shown great potential to improve the efficiency of transportation systems in terms of passenger comfort, fuel economy, stability of driving behavior and mitigation of traffic congestions.
Yet, to deploy these vehicles and leverage their benefits, the underlying algorithms must ensure their safe operation.
In this paper, we address the safety of connected cruise control strategies for longitudinal car following using \emph{control barrier function (CBF)} theory.
In particular, we consider various safety measures such as minimum distance, time headway and time to conflict, and provide a formal analysis of these measures through the lens of CBFs. 
Additionally, motivated by how stability charts facilitate stable controller design, we derive \emph{safety charts} for existing connected cruise controllers to identify safe choices of controller parameters. 
Finally, we combine the analysis of safety measures and the corresponding stability charts to synthesize safety-critical connected cruise controllers using CBFs.
We verify our theoretical results by numerical simulations.
\end{abstract}

\section{INTRODUCTION}
\label{sec:intro}


Vehicle automation holds the promise of improving the efficiency of traffic systems, with great prospective benefits in safety, passenger comfort, fuel economy, mitigation of traffic congestions and reduction of travel times.
The success of automated vehicles (AVs), however, is conditioned on designing efficient longitudinal and lateral controllers.
Therefore, strategies like {\em adaptive cruise control (ACC)} have been studied extensively with great success.
%
The performance of AVs is further improved by providing additional information about the surrounding traffic via vehicle-to-everything (V2X) connectivity---this leads to connected automated vehicles (CAVs) with better ability to respond to other road participants.
For example, {\em cooperative adaptive cruise control (CACC)}~\cite{wang2018review_CACC}
allows platoons of CAVs to share information, cooperate, and improve their driving behavior.
{\em Connected cruise control (CCC)}~\cite{Ge2014}, on the other hand, regulates the motion of a single CAV while leveraging information shared by other connected (but not necessarily automated) vehicles.
With well-designed ACC, CACC or CCC systems, CAVs have shown significant benefits in energy saving~\cite{vahidi2018energy} and in mitigating traffic congestions~\cite{cui2017stabilizing, zheng2020smoothing}.
Remarkably, these benefits have also been showcased by experiments~\cite{Ge2018, stern2018dissipation}.


To deploy CAVs and thereby enjoy their benefits, safe and collision-free behavior is of primary concern.
Recently, the literature has put emphasis on safety-critical control designs for CAVs.
These include safe ACC systems established using reachability analysis~\cite{alam2014guaranteeing}, formal methods~\cite{Niletal2016} and control barrier functions (CBFs)~\cite{ames2014control,AmesXuGriTab2017}, and safe CACC with model predictive control~\cite{massera2017safe}, just to mention a few examples.
Now, we focus on CBF-based approaches, due to their success in a variety of application areas, including AV experiments in traffic~\cite{gunter2022experimental},
AVs executing obstacle avoidance~\cite{chen2018obstacle},
multi-agent systems capturing AVs~\cite{jankovic2022multiagent},
traffic merging~\cite{xiao2019decentralized},
roundabout crossing~\cite{abduljabbar2021cbfbased}, and safe traffic control by CAVs~\cite{zhao2023safetycritical}.
Safe CCC with CBFs has also appeared recently~\cite{He2018}, and it has been implemented on a full-size truck and successfully tested in experiments~\cite{Alan2022AV}.


\begin{figure}
\centering
\includegraphics[scale=1]{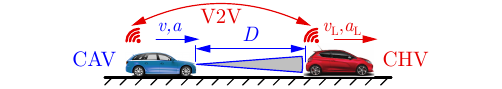}
\caption{
Connected cruise control (CCC) setup: a connected automated vehicle (CAV) is controlled to safely follow a connected human-driven vehicle (CHV) by using information from vehicle-to-vehicle (V2V) communication.
}
\label{fig:carfollowing}
\vspace{-5mm}
\end{figure}

In this paper, we establish safe CCC designs that allow CAVs to follow other vehicles with guaranteed safety w.r.t.~various metrics like minimum distance, time headway or time to conflict.
We make two contributions through the application of CBF theory.
First, we analyze the safety of an existing CCC strategy, and determine provably safe choices of controller parameters.
These results are summarized as safety charts---a concept adopted from~\cite{He2018}.
Second, we synthesize safety-critical controllers by minimally modifying existing, potentially unsafe designs.
We use numerical simulations to verify the theoretical analysis.
Throughout the paper we highlight the benefits of connectivity in order to maintain safety in mixed traffic scenarios.
The results presented are also extendable to other mobile agent systems where spatiotemporal separation between the agents is required, e.g., legged robots, airborne agents and sea vessels.


\section{CONNECTED CRUISE CONTROL}
\label{sec:CCC}

Consider the scenario in Fig.~\ref{fig:carfollowing}, where a connected automated vehicle (CAV) is controlled longitudinally to follow a connected human-driven vehicle (CHV) on a single lane straight road while maintaining a safe distance.
We assume that the CAV has access to measurements of its own speed $v$ and acceleration $a$, the leading CHV's speed $\vL$ and acceleration $\aL$, and the distance $D$, by the help of on-board range sensors and vehicle-to-vehicle (V2V) connectivity.
Note that $\aL$ is typically difficult to obtain by range sensors, while V2V communication can help provide it.

We describe car following by the state ${x = \begin{bmatrix} D &\!\!\! v &\!\!\! \vL \end{bmatrix}^\top}$ and the model:
\begin{align}
\begin{split}
    \dot{D} & = \vL - v, \\
    \dot{v} & = u - p(v), \\
    \dot{v}_{\rm L} & = \aL,
\end{split}
\label{eq:system_nolag}
\end{align}
where the CAV executes the acceleration command $u$ subject to rolling and air resistance captured by ${p(v) \geq 0}$.

The car-following task can be accomplished by the following {\em connected cruise control (CCC)} strategy, ${u = \kd(x)}$, that was proposed in~\cite{Ge2014} and experimentally tested in~\cite{Ge2018}: 
\begin{equation}
    \kd(x) = A \big( V(D) - v \big) + B \big( W(\vL) - v \big) + C \aL.
\label{eq:CCC}
\end{equation}
The CAV responds to the distance, speed difference and CHV's acceleration with gains ${A, B, C \!\geq\! 0}$, respectively, and:
\begin{align}
\begin{split}
    V(D) & = \min\{\kappa(D-\Dst),\vmax\}, \\
    W(\vL) & = \min\{\vL,\vmax\}.
\end{split}
\label{eq:VW}
\end{align}
Here the speed policy $W$ prevents the CAV from following a CHV that exceeds the speed limit $\vmax$, while the range policy $V$ prescribes a desired speed as a function of the distance $D$, that is zero at the standstill distance $\Dst$ and increases linearly with gradient ${\kappa>0}$ up to the speed limit.

Fig.~\ref{fig:simulation} shows the performance of CCC by numerical simulations of~(\ref{eq:system_nolag})-(\ref{eq:CCC}) for two different sets of control gains (solid lines and dash-dot lines).
Unless stated otherwise, the parameters of each numerical result in this paper are those listed in Table~\ref{tab:parameters}.
The simulations\footnote{Matlab codes for each example are available at: \url{https://github.com/molnartamasg/safe-connected-cruise-control}.} capture an emergency braking where the CHV comes to a stop.
Using CCC, the CAV responds to this event and decelerates; see panels (b) and (d).
Notice in panel (a) that the distance between the vehicles greatly depends on the choice of control gains.
The selection of controller parameters has significant impact on safety, that is highlighted in panel (c) and is discussed below.

\begin{figure}
\centering
\includegraphics[scale=1]{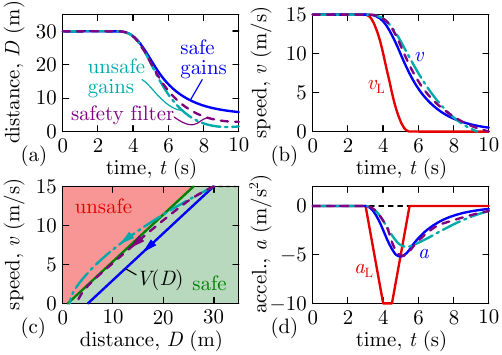}
\caption{
Simulations of model~(\ref{eq:system_nolag}) with CCC~(\ref{eq:CCC}), using a provably safe choice of controller parameters (solid lines) and an unsafe choice (dash-dot lines).
These sets of parameters correspond to points P and Q in Fig.~\ref{fig:safety}(a), respectively.
For the unsafe CCC setup, formal safety guarantees can be recovered by utilizing the safety filter~(\ref{eq:ksTH})-(\ref{eq:QPsolumin}) (dashed lines).
}
\label{fig:simulation}
\end{figure}

\begin{figure}
\centering
\includegraphics[scale=1]{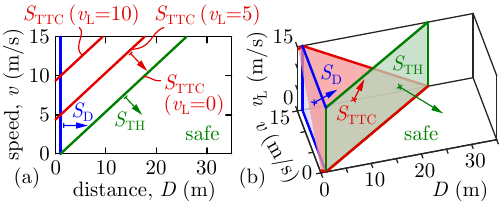}
\caption{
Safe sets of longitudinal car-following, considering safety with respect to distance (D), time headway (TH) and time to conflict (TTC).
Arrows indicate the safe side of the set boundaries.
}
\label{fig:sets}
\vspace{-3mm}
\end{figure}

\bgroup
\setlength{\tabcolsep}{3pt}
\begin{table}
\caption{Parameter Values Used for the Numerical Examples}
\begin{center}
\begin{tabular}{c|c|c|c}
Parameter & Symbol & Value & Unit \\
\hline
resistance terms & $p(v)$ & 0 & m/s$^2$ \\
\hline
standstill distance & $\Dst$ & 5 & m \\
range policy gradient & $\kappa$ & 0.6 & 1/s \\
speed limit & $\vmax$ & 15 & m/s \\
safe control gains (point P) & $(A,B,C)$ & (0.4,0.6,0) & (1/s,1/s,1) \\
unsafe control gains (point Q) & $(A,B,C)$ & (0.4,0.3,0) & (1/s,1/s,1) \\
\hline
safe distance & $\Dsf$ & 1 & m \\
safe time headway & $\TH$ & 1.67 & s \\
safe time to conflict & $\TTC$ & 1.67 & s \\
maximum speed & ${\vmax = \bar{v}}$ & 15 & m/s \\
\end{tabular}
\end{center}
\label{tab:parameters}
\vspace{-3mm}
\end{table}
\egroup

To characterize the safety of longitudinal control, we rely on various safety measures that require different kinds of spatiotemporal separations between the vehicles.
We list three possible safety criteria below, in the form ${h(x) \geq 0}$, associated with a safety measure $h$ and a {\em safe set} $\S$ of states.
\begin{enumerate}[label=(\roman*)]
\item
{\em Distance} must be kept above a safe value $\Dsf$:
\begin{align}
\begin{split}
    \S_{\rm D} & = \{x \in \R^3: D \geq \Dsf\}, \\
    h_{\rm D}(x) & = D-\Dsf.
\end{split}
\label{eq:safesetD}
\end{align}
\item
{\em Time headway} must be kept above a safe value $\TH$:
\begin{align}
\begin{split}
    \S_{\rm TH} & = \{x \in \R^3: D \geq \Dsf + \TH v\}, \\
    h_{\rm TH}(x) & = (D-\Dsf)/\TH - v.
\end{split}
\label{eq:safesetTH}
\end{align}
\item
{\em Time to conflict} must be kept above a safe value $\TTC$:
\begin{align}
\begin{split}
    \S_{\rm TTC} & = \{x \in \R^3: D \geq \Dsf + \TTC (v-\vL)\}, \\
    h_{\rm TTC}(x) & = (D-\Dsf)/\TTC + \vL-v.
\end{split}
\label{eq:safesetTTC}
\end{align}
\end{enumerate}
The time to conflict is often referred to as time to collision if ${\Dsf=0}$.
For further choices of safety indicators, please see \cite{zhao2023safetycritical, He2018} and the references therein.

The safe sets in~(\ref{eq:safesetD})-(\ref{eq:safesetTTC}) are depicted in Fig.~\ref{fig:sets}, where their boundaries are shown by thick lines in panel (a) and as planes in panel (b).
The system is safe w.r.t.~the given safety measure if it evolves in the half space indicated by arrows.
Note that the time headway is the strictest of the three safety indicators: it requires the largest distance at a given speed to be safe.
The safety of the previous simulation results w.r.t.~the time headway is evaluated in Fig.~\ref{fig:simulation}(c).
While CCC~(\ref{eq:CCC}) keeps system~(\ref{eq:system_nolag}) safe with one choice of CCC parameters (solid lines), it fails to do so with another (dash-dot lines).
As such, choosing the parameters of the CAV's controller has crucial impact on safety and safe parameters must be identified.

The problem of choosing controller parameters has been well-studied for CCC from stability perspective.
Specifically,~\cite{Ge2014} has derived {\em stability charts} for~(\ref{eq:system_nolag})-(\ref{eq:CCC}) that identify controller parameters associated with stable driving.
These charts distinguish the {\em plant stable} domain, where the CAV is able to approach a constant speed in a stable way, and the {\em string stable} region, where the CAV has smaller speed fluctuations than the CHV, thereby smoothing the traffic flow.
Examples of these stability charts from~\cite{Ge2014} are plotted in Fig.~\ref{fig:stability}, where the plant stable domain, given by ${A \geq 0}$ and ${A \geq -B}$, is red, and the string stable region, given by ${A \geq 0}$, ${A \geq 2 \big( (1-C)\kappa -B \big)}$ and ${C \leq 1}$, is blue.
The charts are shown for ${C=0}$ in panel (a), and ${C = 0, 0.25, 0.75}$ (thin lines), ${C=0.5}$ (thick line and shading) and ${C \to 1}$ (dashed line) in panel (b).
In what follows, we establish {\em safety charts}---whose concept first appeared in~\cite{He2018}---in a similar way to identify safe controller parameters.
As preliminary, we revisit the theory of control barrier functions.

\begin{figure}
\centering
\includegraphics[scale=1]{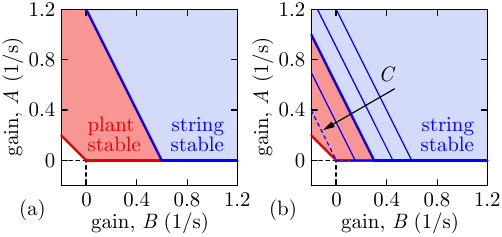}
\caption{
Stability chart~\cite{Ge2014} of CCC~(\ref{eq:system_nolag})-(\ref{eq:CCC}) with (a) ${C=0}$, (b) various ${C>0}$ values.
}
\label{fig:stability}
\vspace{-3mm}
\end{figure}


\section{CONTROL BARRIER FUNCTIONS}
\label{sec:CBF}

Consider a control system with state ${x \in \X}$, input ${u \in \U}$ and dynamics given by locally Lipschitz continuous functions ${f: \X \to \X}$ and ${g: \X \to \R^{n \times m}}$:
\begin{equation}
    \dot{x} = f(x) + g(x) u;
\label{eq:system}
\end{equation}
cf.~(\ref{eq:system_nolag}) with
${f(x) = \begin{bmatrix} \vL \!-\! v &\!\!\! - p(v) &\!\!\! \aL \end{bmatrix}^{\!\top\!}}$,
${g(x) = \begin{bmatrix} 0 &\!\!\! 1 &\!\!\! 0 \end{bmatrix}^{\!\top\!}}$.
With a locally Lipschitz continuous controller ${k: \X \to \U}$, ${u = k(x)}$, such as~(\ref{eq:CCC}), the corresponding closed loop system:
\begin{equation}
    \dot{x} = f(x) + g(x) k(x),
\label{eq:closedloop}
\end{equation}
with the initial condition ${x(0) = x_0 \in \X}$, has a unique solution $x(t)$, which we assume to exist for all ${t \geq 0}$.

We call~(\ref{eq:closedloop}) as safe if its solution $x(t)$ evolves within a safe set $\S$ for all time.
We consider $\S$, and its boundary $\bS$, to be given by a continuously differentiable function ${h: \X \to \R}$:
\begin{align}
\begin{split}
    \S & = \{x \in \X: h(x) \geq 0 \}, \\
    \bS & = \{x \in \X: h(x) = 0 \}.
\end{split}
\label{eq:safeset}
\end{align}
That is, ${h(x(t)) \geq 0}$ for all ${t \geq 0}$ indicates safety while ${h(x(t)) < 0}$ for any ${t \geq 0}$ is unsafe, where $h$ is selected based on the application; cf.~(\ref{eq:safesetD}),~(\ref{eq:safesetTH}) and~(\ref{eq:safesetTTC}).
To maintain ${h(x(t)) \geq 0}$, we rely on the derivative of $h$ along~(\ref{eq:system}):
\begin{equation}
    \dot{h}(x,u) =
    \underbrace{\grad{h}(x) f(x)}_{\L{f}{h(x)}}
    + \underbrace{\grad{h}(x) g(x)}_{\L{g}{h(x)}} u.
\label{eq:hdot}
\end{equation}
With this, Nagumo's theorem~\cite{nagumo1942lage} establishes safety for~(\ref{eq:closedloop}).

\begin{theorem}[\cite{nagumo1942lage}] \label{theo:Nagumo}
\textit{
Let $h$ satisfy ${\grad{h}(x) \neq 0}$, ${\forall x \in \bS}$.
System~(\ref{eq:closedloop}) is safe w.r.t.~$\S$ if and only if:
\begin{equation}
    \dot{h} \big( x,k(x) \big) \geq 0, \quad \forall x \in \bS.
\label{eq:Nagumo}
\end{equation}
}
\end{theorem}

Condition~(\ref{eq:Nagumo}) means that the controller does not allow the system to leave the safe set $\S$ when it is at the boundary $\bS$.
To {\em certify} that~(\ref{eq:closedloop}) with a given controller $k$ is safe, one needs to verify that~(\ref{eq:Nagumo}) holds.
Yet,~(\ref{eq:Nagumo}) does not provide a constructive way to {\em synthesize} controllers for~(\ref{eq:system}), since it does not provide guidelines inside $\S$ (i.e., when ${h(x) > 0}$).

Control barrier functions (CBFs)~\cite{AmesXuGriTab2017} have been proposed for the purpose of safety-critical controller synthesis.

\begin{definition}[\cite{AmesXuGriTab2017}]
Function $h$ is a {\em control barrier function} for~(\ref{eq:system}) on $\S$ if there exists ${\alpha \in \Keinf}$\footnote{Function ${\alpha : \R \to \R}$ is of extended class-$\Kinf$ (${\alpha \in \Keinf}$) if it is continuous, strictly increasing, ${\alpha(0)=0}$ and ${\lim_{r \to \pm \infty} \alpha(r) = \pm \infty}$.} such that for all ${x \in \S}$:
\begin{equation}
    \sup_{u \in \U} \dot{h}(x,u) > - \alpha \big( h(x) \big).
\label{eq:CBF_condition}
\end{equation}
\end{definition}

Note that the $\sup$ on left-hand side of~(\ref{eq:CBF_condition}) gives $\L{f}{h(x)}$ if ${\L{g}{h(x)} = 0}$ and $\infty$ otherwise, thus~(\ref{eq:CBF_condition}) is equivalent to:
\begin{equation}
    \L{f}{h(x)} > - \alpha \big( h(x) \big), \quad \forall x \in \X \ {\rm s.t.}\ \L{g}{h(x)} = 0.
\label{eq:CBF_condition_rewritten}
\end{equation}
\cite{AmesXuGriTab2017} established safety-critical control with CBFs as follows.

\begin{theorem}[\cite{AmesXuGriTab2017}] \label{theo:CBF}
\textit{
If $h$ is a CBF for~(\ref{eq:system}) on $\S$, then any locally Lipschitz continuous controller $k$ that satisfies: 
\begin{equation}
    \dot{h} \big( x, k(x) \big) \geq - \alpha \big( h(x) \big)
\label{eq:safety_condition}
\end{equation}
for all ${x \in \S}$ renders~(\ref{eq:closedloop}) safe w.r.t.~$\S$.
}
\end{theorem}

Note that if~(\ref{eq:safety_condition}) holds, then~(\ref{eq:Nagumo}) also does.
Furthermore, condition~(\ref{eq:safety_condition}) provides guidelines over the entire set $\S$ to synthesize controllers.
For example, CBFs are often used in {\em safety filters} that modify a desired but not necessarily safe controller ${\kd: \X \to \U}$ to a safe controller subject to~(\ref{eq:safety_condition}), in the form of an optimization problem (quadratic program):
\begin{align}
\begin{split}
    k(x) = \underset{u \in \U}{\operatorname{argmin}} & \quad \| u - \kd(x) \|^2 \\
    \text{s.t.} & \quad \dot{h}(x,u) \geq - \alpha \big( h(x) \big).
\end{split}
\label{eq:QP}
\end{align}
The solution of~(\ref{eq:QP}) can be given in closed form~\cite{Alan2022AV}:
\begin{align}
    k(x) & = \begin{cases}
        \kd(x) + \max\{0, \eta(x) \} \frac{\L{g}{h(x)}^\top}{\|\L{g}{h(x)}\|^2}, & {\rm if}\ \L{g}{h(x)} \neq 0, \\
        \kd(x), & {\rm if}\ \L{g}{h(x)} = 0,
    \end{cases} \nonumber \\
    \eta(x) & = -\L{f}{h(x)} - \L{g}{h(x)} \kd(x) - \alpha \big( h(x) \big).
\label{eq:QPsolu}
\end{align}

For single input systems like~(\ref{eq:system_nolag}), where $u$ and $\L{g}{h(x)}$ are scalars, the safety conditions greatly simplify.
If ${\L{g}{h(x)} < 0}$,~(\ref{eq:safety_condition}) is equivalent to:
\begin{equation}
    k(x) \leq \ks(x),
\label{eq:safety_condition_Lghneg}
\end{equation}
with:
\begin{equation}
    \ks(x) = -\frac{\L{f}{h(x)} + \alpha \big( h(x) \big)}{\L{g}{h(x)}}.
\label{eq:ks}
\end{equation}
If ${\L{g}{h(x)} > 0}$,~(\ref{eq:safety_condition}) yields:
\begin{equation}
    k(x) \geq \ks(x).
\label{eq:safety_condition_Lghpos}
\end{equation}
If ${\L{g}{h(x)} = 0}$,~(\ref{eq:CBF_condition_rewritten}) guarantees that~(\ref{eq:safety_condition}) holds for any ${k(x)}$.
Thus, for scalar input $u$, the safety filter~(\ref{eq:QPsolu}) becomes~\cite{Alan2022AV}:
\begin{equation}
    k(x) =
    \begin{cases}
        \min\{\kd(x),\ks(x)\}, & {\rm if}\ \L{g}{h(x)} < 0, \\
        \kd(x), & {\rm if}\ \L{g}{h(x)} = 0, \\
        \max\{\kd(x),\ks(x)\}, & {\rm if}\ \L{g}{h(x)} > 0.
    \end{cases}
\label{eq:QPsoluscalar}
\end{equation}

Finally, it is important to distinguish the case ${\L{g}{h(x)} \equiv 0}$, where the input $u$ does not affect safety directly in~(\ref{eq:hdot}) for any $x$.
Then, $h$ is not a CBF and safety-critical controller synthesis is not possible directly with $h$ (unless~(\ref{eq:closedloop}) is safe for any $k(x)$).
Instead, one may construct an {\em extended CBF}~\cite{Nguyen2016, xiao2019cbf,ames2020integral} with a continuously differentiable ${\alpha \in \Keinf}$: 
\begin{equation}
    \he(x) = \L{f}{h(x)} + \alpha \big( h(x) \big),
\label{eq:CBFextension}
\end{equation}
that is associated with the {\em extended safe set}:
\begin{align}
\begin{split}
    \Se & = \{x \in \X: \he(x) \geq 0 \}, \\
    \bSe & = \{x \in \X: \he(x) = 0 \}.
\end{split}
\label{eq:safeset_extended}
\end{align}
If the system is kept inside $\Se$, condition~(\ref{eq:safety_condition}) holds, and the system also evolves within $\S$.
Ultimately, safety w.r.t.~the intersection ${\S \cap \Se}$ of the two sets is guaranteed as follows.

\begin{corollary}[\cite{xiao2019cbf}] \label{cor:extendedCBF}
\textit{
If ${\L{g}{h(x)} \equiv 0}$ and $\he$ in~(\ref{eq:CBFextension}) is a CBF for~(\ref{eq:system}) on $\Se$ with ${\alphae \in \Keinf}$, then any locally Lipschitz continuous controller $k$ that satisfies: 
\begin{equation}
    \dot{h}_{\rm e} \big( x, k(x) \big) \geq - \alphae \big( \he(x) \big)
\label{eq:safety_condition_extended}
\end{equation}
for all ${x \in \Se}$ renders~(\ref{eq:closedloop}) safe w.r.t.~${\S \cap \Se}$.
}
\end{corollary}
\noindent With this result, safety filters incorporating~(\ref{eq:safety_condition_extended}) can be constructed analogously to~(\ref{eq:QP})-(\ref{eq:QPsolu}), with $\he$ instead of $h$.

Accordingly, safety certification by Nagumo's theorem---the extension of Theorem~\ref{theo:Nagumo} for ${\L{g}{h(x)} \equiv 0}$---is performed on the boundary of ${\S \cap \Se}$.
This boundary is located at ${\bSe \cap \S}$ (where ${\he(x) = 0}$ and ${h(x) \geq 0}$) and at ${\bS \cap \Se}$ (where ${h(x) = 0}$ and ${\he(x) \geq 0}$).
Note that only the former case needs further analysis, since the latter case implies~${\he(x) = \dot{h} \big( x,k(x) \big) \geq 0}$, and the system cannot leave ${\S \cap \Se}$ along this boundary per Theorem~\ref{theo:Nagumo}.
Thus, safety certification is summarized as follows.
\begin{corollary} \label{cor:extendedNagumo}
\textit{
Let ${\L{g}{h(x)} \equiv 0}$ and $\he$ in~(\ref{eq:CBFextension}) satisfy ${\grad{\he}(x) \!\neq\! 0}$, ${\forall x \in \bSe}$.
System~(\ref{eq:closedloop}) is safe w.r.t.~${\S \!\cap \!\Se}$ if:
\begin{equation}
    \dot{h}_{\rm e} \big( x,k(x) \big) \geq 0, \quad \forall x \in \bSe \cap \S.
\label{eq:Nagumo_extended}
\end{equation}
}
\end{corollary}

Again, notice that~(\ref{eq:Nagumo_extended}) holds if~(\ref{eq:safety_condition_extended}) does.
In the case of a single input,~(\ref{eq:safety_condition_extended}) as well as the corresponding safety filter can be expressed in simple form.
Analogously to the non-extended case, formulas~(\ref{eq:safety_condition_Lghneg}),~(\ref{eq:safety_condition_Lghpos}) and~(\ref{eq:QPsoluscalar}) can be used with:
\begin{equation}
    \ks(x) = -\frac{\L{f}{\he(x)} + \alphae \big( \he(x) \big)}{\L{g}{\he(x)}},
\label{eq:ks_extended}
\end{equation}
cf.~(\ref{eq:ks}).


\section{SAFE CONNECTED CRUISE CONTROL}
\label{sec:safeCCC}

Now, we apply CBF theory to analyze the safety of CCC~(\ref{eq:system_nolag})-(\ref{eq:CCC}), and to synthesize safety-critical CCC laws.

\subsection{Safe Time Headway}

First, we characterize the safety of CCC w.r.t.~the time headway criterion in~(\ref{eq:safesetTH}).
For brevity, we introduce
${\bar{\kappa} = 1/\TH}$.
The gradient and Lie derivatives of $h_{\rm TH}$ in~(\ref{eq:hdot}) read:
\begin{align}
\begin{split}
    \grad{h_{\rm TH}}(x) & = \begin{bmatrix} \bar{\kappa} & - 1 & 0 \end{bmatrix}, \quad
    \L{g}{h_{\rm TH}(x)} = - 1, \\
    \L{f}{h_{\rm TH}(x)} & = \bar{\kappa} (\vL - v) + p(v).
\end{split}
\end{align}
Note that both
${\grad{h_{\rm TH}}(x) \neq 0}$, ${\forall x \in \partial \S_{\rm TH}}$
and~(\ref{eq:CBF_condition_rewritten}) hold, hence Theorems~\ref{theo:Nagumo} and~\ref{theo:CBF} are applicable for certifying safety and synthesizing safety-critical controllers, respectively.

We first use Theorem~\ref{theo:Nagumo} to analyze the safety of the CCC introduced in~(\ref{eq:CCC}), i.e., ${u = k(x) = \kd(x)}$.
Since the input is scalar and ${\L{g}{h_{\rm TH}(x)} < 0}$,
~(\ref{eq:Nagumo}) is equivalent to:
\begin{equation}
    \ks(x) - \kd(x) \geq 0, \quad \forall x \in \X \ {\rm s.t.}\ h_{\rm TH}(x) = 0,
\label{eq:Nagumo_special}
\end{equation}
cf.~(\ref{eq:safety_condition_Lghneg}), where the safe input defined in~(\ref{eq:ks}) is given by:
\begin{equation}
    \ks(x) = \bar{\kappa}(\vL-v) + p(v) + \alpha \big( \bar{\kappa}(D-\Dsf) - v \big).
\label{eq:ksTH}
\end{equation}
By analyzing under what conditions~(\ref{eq:Nagumo_special}) holds, we arrive at the following result.
For the detailed analysis and the proof of this result, please see the Appendix.
\begin{theorem} \label{theo:nolag_TH}
\textit{
System~(\ref{eq:system_nolag}) with ${u \!=\! k(x) \!=\! \kd(x)}$ given by~(\ref{eq:CCC}), ${A, B \!\geq\! 0}$ and ${C=0}$ is safe w.r.t.~$\S_{\rm TH}$ in~(\ref{eq:safesetTH}) if:
\begin{itemize}
\item
${v \geq 0}$, ${\Dst \geq \Dsf}$ and
${B = \bar{\kappa} \geq \kappa}$; or
\item
${v,\vL \in [0,\bar{v}]}$ with some ${\bar{v} \geq 0}$,
${\Dst>\Dsf}$,
${\bar{\kappa} \geq \kappa}$
and:
\begin{equation}
    A \geq \frac{|\bar{\kappa}-B| \bar{v}}{\kappa (\Dst - \Dsf)}.
\label{eq:TH_safety_nolag}
\end{equation}
\end{itemize}
}
\end{theorem}

\begin{figure}
\centering
\includegraphics[scale=1]{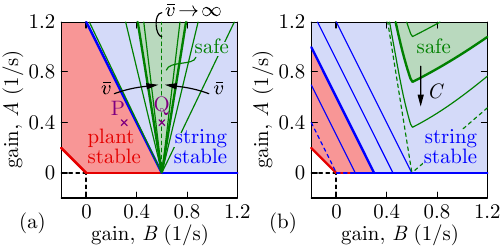}
\caption{
Safety charts of CCC~(\ref{eq:system_nolag})-(\ref{eq:CCC}) (a) w.r.t.~the time headway criterion~(\ref{eq:safesetTH}) with ${C=0}$ and various maximum speed $\bar{v}$;
(b) w.r.t.~the distance and time to conflict in~(\ref{eq:safesetD}) and (\ref{eq:safesetTTC}), respectively, with various $C$.
}
\label{fig:safety}
\vspace{-3mm}
\end{figure}

Condition~(\ref{eq:TH_safety_nolag}) can be visualized in the space $(B,A)$ of control gains, resulting in the {\em safety chart} in Fig.~\ref{fig:safety}(a).
The safe domain---associated with a provably safe choice of control gains---is shown for ${\bar{\kappa} = \kappa}$ and the maximum speed ${\bar{v}=15\,{\rm m/s}}$ with thick green boundary and green shading, on top of the stable domains (red and blue) that were plotted in Fig.~\ref{fig:stability}(a).
Additional boundaries are shown for ${\bar{v}=5, 10}$ and ${20\,{\rm m/s}}$ (thin lines) and the limit ${\bar{v} \to \infty}$ (dashed line).
As the maximum speed $\bar{v}$ increases, the V-shaped safe region closes to a single line given by the first bullet point in Theorem~\ref{theo:nolag_TH}.
The safe and unsafe simulations in Fig.~\ref{fig:simulation} correspond to points P and Q, respectively, that indeed lie in the safe and unsafe domains.
Note that safety w.r.t.~time headway can be achieved even without acceleration feedback (${C=0}$), while the response to acceleration (via ${C \aL}$ in~(\ref{eq:CCC})) will be necessary for safety w.r.t.~distance and time-to-collision.

The safety charts allow us to select the parameters of the CCC~(\ref{eq:CCC}) in a safe way.
Alternatively, one can also synthesize a safety-critical controller via Theorem~\ref{theo:CBF}, by viewing~(\ref{eq:CCC}) as desired input and using a CBF-based safety filter.
Since ${\L{g}{h_{\rm TH}(x)} < 0}$, the safety filter~(\ref{eq:QPsoluscalar}) simplifies to:
\begin{equation}
    k(x) = \min\{\kd(x), \ks(x)\}.
\label{eq:QPsolumin}
\end{equation}
The behavior of the safety filter in~(\ref{eq:ksTH}) and~(\ref{eq:QPsolumin}) is demonstrated in Fig.~\ref{fig:simulation} by dashed lines.
The safety filter is applied on the desired controller~(\ref{eq:CCC}) with the unsafe gains corresponding to the dash-dot lines (cf.~point Q in Fig.~\ref{fig:safety}(a)) and ${\alpha(r)=r}$.
The end result is provably safe CCC.

In conclusion, safety can be guaranteed both by controller tuning through safety charts and by applying safety filters.
Safety charts combined with stability charts (or other analysis on performance) provide both safe and performant controllers.
However, if safe regions are too small, or do not overlap with stable regions, one may design CCC based on performance only (i.e., properties like string stability), and apply safety filters.
Safety filters ensure safety even when safe parameters for nominal CCC laws are hard to realize.


\subsection{Safe Distance and Time to Conflict}
Next, we address safety w.r.t.~distance, as in~(\ref{eq:safesetD}), for which:
\begin{align}
\begin{split}
    \grad{h_{\rm D}}(x) & =
    \begin{bmatrix}
    1 & 0 & 0
    \end{bmatrix}, \quad
    \L{g}{h_{\rm D}(x)} \equiv 0, \\
    \L{f}{h_{\rm D}(x)} & = \vL - v.
\end{split}
\end{align}
Since ${\L{g}{h_{\rm D}(x)} \equiv 0}$, an extended CBF must be constructed from $h_{\rm D}$ via~(\ref{eq:CBFextension}).
We do this by the following observation.
\begin{observation}
For system~(\ref{eq:system_nolag}), the time to conflict-based safety measure $h_{\rm TTC}$ in~(\ref{eq:safesetTTC}) can be expressed using the distance-based safety indicator $h_{\rm D}$ in~(\ref{eq:safesetD}) as:
\begin{equation}
    h_{\rm TTC}(x) = \L{f}{h_{\rm D}(x)} + h_{\rm D}(x)/\TTC.
\end{equation}
That is, $h_{\rm TTC}$ is an extension~(\ref{eq:CBFextension}) of $h_{\rm D}$ with 
${\alpha(r) = r/\TTC}$.
\end{observation}
Thus, $h_{\rm TTC}$ is regarded as extended CBF $\he$, yielding:
\begin{align}
\begin{split}
    \grad{h_{\rm TTC}}(x) & =
    \begin{bmatrix}
    \bar{\kappa} & - 1 & 1
    \end{bmatrix}, \quad
    \L{g}{h_{\rm TTC}(x)} = - 1, \\
    \L{f}{h_{\rm TTC}(x)} & = \bar{\kappa} (\vL - v) + \aL + p(v),
\end{split}
\end{align}
where ${\bar{\kappa} = 1/\TTC}$.
Observe that
${\grad{h_{\rm TTC}}(x) \neq 0}$, ${\forall x \in \partial \S_{\rm TTC}}$
holds and $h_{\rm TTC}$ is in fact a CBF.
Thus, Corollaries~\ref{theo:Nagumo} and~\ref{theo:CBF} apply, and safety is established w.r.t.~${\S_{\rm D} \cap \S_{\rm TTC}}$ as follows.
\begin{proposition}
\textit{
System~(\ref{eq:system_nolag}) is safe w.r.t.~${\S_{\rm D} \cap \S_{\rm TTC}}$ given by (\ref{eq:safesetD}) and (\ref{eq:safesetTTC}) if~(\ref{eq:Nagumo_extended}) holds (with ${\he(x) = h_{\rm TTC}(x)}$) for a given controller, ${u = k(x)}$.
Moreover, any controller, ${u = k(x)}$, that satisfies~(\ref{eq:safety_condition_extended}) renders~(\ref{eq:system_nolag}) safe w.r.t.~${\S_{\rm D} \cap \S_{\rm TTC}}$.
}
\end{proposition}
Thus, safe distance is guaranteed by ensuring safe time to conflict, and safety is ultimately achieved for ${\S_{\rm D} \cap \S_{\rm TTC}}$.
Moreover, note that if the vehicles do not move in reverse, i.e., ${v, \vL \geq 0}$, and if ${\TH = \TTC}$, then
${h_{\rm TH}(x) \geq 0}$ implies ${h_{\rm D}(x) \geq 0}$ and ${h_{\rm TTC}(x) \geq 0}$; cf.~(\ref{eq:safesetD})-(\ref{eq:safesetTTC}).
This means that ensuring safety w.r.t.~time headway yields safety w.r.t.~time to conflict and~distance (provided that ${h_{\rm TH}(x_0) \geq 0}$ holds).

The safety of~(\ref{eq:system_nolag}) w.r.t.~distance and time to conflict with the CCC law in~(\ref{eq:CCC}) can be certified by~(\ref{eq:Nagumo_extended}) that reduces to:
\begin{equation}
    \ks(x) - \kd(x) \!\geq\! 0, \;\, \forall x \!\in\! \X \ {\rm s.t.}\ h_{\rm TTC}(x) \!=\! 0\ {\rm and}\ h_{\rm D}(x) \!\geq\! 0,
\label{eq:Nagumo_extended_special}
\end{equation}
as ${\L{g}{h_{\rm TTC}(x)} < 0}$, cf.~(\ref{eq:safety_condition_Lghneg}).
Here $\ks$ is obtained from~(\ref{eq:ks_extended}):
\begin{equation}
    \ks(x) = \bar{\kappa}(\vL \!-\! v) + p(v) + \alphae \big( \bar{\kappa}(D \!-\! \Dsf) + \vL-v \big) + \aL.
\label{eq:ksTTC}
\end{equation}
After analyzing for which parameters~(\ref{eq:Nagumo_extended_special}) holds, we obtain the following result (with proof in the Appendix).

\begin{theorem} \label{theo:nolag_TTC}
\textit{
System~(\ref{eq:system_nolag}) with ${u \!=\! k(x) \!=\! \kd(x)}$ given by~(\ref{eq:CCC}) and ${A, B, C \geq 0}$ is safe w.r.t.~${\S_{\rm D} \cap \S_{\rm TTC}}$ in~(\ref{eq:safesetD}) and~(\ref{eq:safesetTTC}) if
${\aL \geq -\gamma(\vL)}$ with some ${\gamma \in \K}$,
${v,\vL \in [0,\bar{v}]}$ with some ${\bar{v} \geq 0}$,
${\Dst>\Dsf}$,
${\bar{\kappa} \geq \kappa}$,
${C \leq 1}$ and:
\begin{multline}
    A \kappa (\Dst-\Dsf) + \min \{0,B-\bar{\kappa}\} \bar{v} \\
    + \min_{\vL \in [0,\bar{v}]} \Big[ (\bar{\kappa}-B+A) \vL - (1-C) \gamma(\vL) \Big] \geq 0.
\label{eq:TTC_safety_nolag}
\end{multline}
}
\end{theorem}

Notice that the condition ${\dot{v}_{\rm L} = \aL \geq -\gamma(\vL)}$ describes the lead CHV's motion and, according to CBF theory, it guarantees ${\vL(0) \geq 0 \implies \vL(t) \geq 0}$, ${\forall t \geq 0}$.
That is, this condition describes that the lead vehicle neither brakes too hard nor drives in reverse.
How ``hard'' it brakes is characterized by the function $\gamma$.
For example, for the constant-jerk profile of the lead CHV in Fig.~\ref{fig:simulation}, it can be shown that ${\aL(t) \geq -\sqrt{20 \vL(t)}}$ for all time, i.e., ${\gamma(r) = \sqrt{20 r}}$.

Similar to~(\ref{eq:TH_safety_nolag}), condition~(\ref{eq:TTC_safety_nolag}) can be visualized in the $(B,A)$ space as safety chart; see Fig.~\ref{fig:safety}(b) for ${\gamma(r) = \sqrt{20 r}}$ and ${\bar{\kappa} = \kappa}$.
The same parameters are used as in Fig.~\ref{fig:stability}: ${C=0}$ in panel (a), and ${C = 0, 0.25, 0.75}$ (thin lines), ${C=0.5}$ (thick lines and shading) and ${C \to 1}$ (dashed lines) in panel (b).
The safe region was found by brute-force evaluation of~(\ref{eq:TTC_safety_nolag}) on a grid of $\vL$, $A$ and $B$ (although expressing $A$ from~(\ref{eq:TTC_safety_nolag}) explicitly could be possible depending on the form of $\gamma$).
The safe region has V-shape, similar to Fig.~\ref{fig:safety}(a), and it moves towards smaller gain $A$ as the acceleration gain $C$ is increased.
This shows that acceleration $\aL$ feedback---that is typically obtained by V2V connectivity---is helpful in achieving safety w.r.t.~distance and time to conflict, since safety would otherwise require large gains and control inputs.

Moreover, apart from the safety charts, the safety filter~(\ref{eq:QPsolumin}) provides an alternative way of safety-critical control regardless of the parameters of CCC~(\ref{eq:CCC}).
Choosing between safety chart and safety filter is up to the  user---the end result is CCC with formal safety guarantees in both cases.

\section*{CONCLUSION}

In this paper, we investigated the safety of connected automated vehicles (CAVs) executing connected cruise control (CCC) by means of control barrier function (CBF) theory.
We established safety charts for existing CCC designs to identify provably safe choices of controller parameters, analogously to stability charts found in the literature.
To recover formal safety guarantees for unsafe parameter choices, we also proposed CBF-based safety filters for controller synthesis.
As future research, we plan to investigate safe CCC in connected vehicle networks where CAVs respond to multiple vehicles.


\section*{APPENDIX}

\begin{proof}[Proof of Theorem~\ref{theo:nolag_TH}]
To prove safety, we apply Theorem~\ref{theo:Nagumo} by showing that~(\ref{eq:Nagumo_special}) holds. We express ${\ks(x)-\kd(x)}$:
\begin{multline}
    \ks(x) - \kd(x) = \alpha(\bar{\kappa}(D-\Dsf) - v) + \bar{\kappa} (\vL-v) + p(v) \\
    - A (V(D) - v) - B (W(\vL)-v),
\end{multline}
and use ${V(D) \leq \kappa(D-\Dst)}$, ${W(\vL) \leq \vL}$, and ${p(v) \geq 0}$:
\begin{multline}
    \ks(x) - \kd(x) \geq \alpha(\bar{\kappa}(D-\Dsf) - v) + \bar{\kappa} (\vL-v) \\
    - A \big( \kappa (D-\Dst) - v \big) - B (\vL-v).
\label{eq:kskd_TH_1}
\end{multline}
This means that providing safety without considering the saturations at $\vmax$ in $V$, $W$ and the resistance term $p(v)$ implies safety with those terms too.
We substitute ${h_{\rm TH}(x) = 0}$ into~(\ref{eq:kskd_TH_1}), which makes the term of $\alpha$ zero; cf.~\eqref{eq:safesetTH}.
Then we add ${A h_{\rm TH}(x) = 0}$ to both sides, and reorganize to: 
\begin{multline}
    \ks(x) - \kd(x) \geq A (\bar{\kappa}-\kappa) (D-\Dsf) + A \kappa (\Dst-\Dsf) \\
    + (\bar{\kappa}-B)(\vL-v).
\label{eq:kskd_TH_2}
\end{multline}
If ${v \geq 0}$, then ${D \!\geq\! \Dsf}$ when ${h_{\rm TH}(x)=0}$. Thus, if ${\Dst \!\geq\! \Dsf}$ and ${B = \bar{\kappa} \geq \kappa}$ also hold,~(\ref{eq:Nagumo_special}) follows and safety is proven.
Furthermore, if ${v,\vL \in [0,\bar{v}]}$, we have ${|\vL-v| \leq \bar{v}}$ and ${D \geq \Dsf}$ when ${h_{\rm TH}(x)=0}$.
With ${\bar{\kappa} \geq \kappa}$,~(\ref{eq:kskd_TH_2}) leads to:
\begin{equation}
    \ks(x) - \kd(x) \geq A \kappa (\Dst-\Dsf) - |\bar{\kappa}-B| \bar{v}.
\end{equation}
Thus,~(\ref{eq:Nagumo_special}) and safety follows for ${\Dst>\Dsf}$ and~(\ref{eq:TH_safety_nolag}).
\end{proof}

\begin{proof}[Proof of Theorem~\ref{theo:nolag_TTC}]
We prove safety by applying Corollary~\ref{cor:extendedNagumo} and showing that~(\ref{eq:Nagumo_extended_special}) holds, where:
\begin{multline}
    \ks(x) - \kd(x) = \alphae \big( \bar{\kappa}(D-\Dsf) + \vL-v \big) + \bar{\kappa} (\vL-v) + \aL \\
    + p(v) - A \big( V(D) - v \big) - B \big( W(\vL) - v \big) - C \aL.
\end{multline}
We use ${V(D) \leq \kappa(D-\Dst)}$, ${W(\vL) \leq \vL}$, and ${p(v) \geq 0}$, substitute ${h_{\rm TTC}(x) = 0}$, which makes the term of $\alphae$ zero; cf.~\eqref{eq:safesetTTC}.
Then we add ${A h_{\rm TTC}(x) = 0}$ to both sides:
\begin{multline}
    \ks(x) - \kd(x) \geq A (\bar{\kappa}-\kappa) (D-\Dsf) + A \kappa (\Dst-\Dsf) \\
    + A \vL + (1-C) \aL + (\bar{\kappa}-B)(\vL-v).
\end{multline}
With ${\aL \!\geq\! -\gamma(\vL)}$, ${C \leq 1}$, ${\kappa \leq \bar{\kappa}}$ and ${h_{\rm D}(x) \!=\! D \!-\! \Dsf \geq 0}$:
\begin{multline}
    \ks(x) - \kd(x) \geq A \kappa (\Dst-\Dsf) + (B-\bar{\kappa}) v \\
    + (\bar{\kappa}-B+A) \vL - (1-C) \gamma(\vL).
\end{multline}
For ${v,\vL \in [0,\bar{v}]}$, we have ${(B-\bar{\kappa}) v \geq \min\{0,B-\bar{\kappa}\} \bar{v}}$, whereas ${\Dst>\Dsf}$ and~(\ref{eq:TTC_safety_nolag}) yield~(\ref{eq:Nagumo_extended_special}) and imply safety.
\end{proof}

% ...


\bibliographystyle{IEEEtran}
\bibliography{2023_cdc}	

\end{document}